\documentclass{llncs}
\usepackage{latexsym}
\usepackage[pdfauthor={Berry Schoenmakers},pdftitle={Amortized Analysis of Leftist Heaps},bookmarks=false]{hyperref}

\title{
\vspace*{-3cm}
\mbox{\hspace*{-5mm}{\tiny In {\em Principles of Verification: Cycling the Probabilistic Landscape. Essays Dedicated to Joost-Pieter Katoen on the Occasion of His 60th Birthday},}} \\[-3mm]
\mbox{\hspace*{-5mm}{\tiny N.\ Hansen et al.\ (eds.), Part I, LNCS 15260, pp.\ 73-84, Springer, 2024. Full version (extended with Appendix~B).}}\\
\vspace*{5mm}
Amortized Analysis of Leftist Heaps}
\author{Berry Schoenmakers\orcidID{0000-0001-6273-8930} \\ \email{berry@win.tue.nl}}
\institute{Math \& CS, TU Eindhoven, Netherlands}

\newcommand{\EMP}{{\sf empty}}
\newcommand{\ISE}[1]{\F{{\sf isempty}}{#1}}
\newcommand{\MNM}[1]{\F{{\sf min}}{#1}}
\newcommand{\INS}[2]{\F{\F{{\sf insert}}{#1}}{#2}}
\newcommand{\DLM}[1]{\F{{\sf delmin}}{#1}}
\newcommand{\SIN}[1]{\F{{\sf single}}{#1}}
\newcommand{\UNI}[2]{\F{\F{{\sf meld}}{#1}}{#2}}
\newcommand{\bal}[1]{\F{\rm bal}{#1}}
\newcommand{\f}[1]{\relax\ifmmode F_{#1}\else $F_{#1}$\fi}
\renewcommand{\l}[1]{\relax\ifmmode L\ifne{#1}{(}#1\ifne{#1}{)}\else $L\ifne{#1}{(}#1\ifne{#1}{)}$\fi}
\renewcommand{\r}[1]{\relax\ifmmode R\ifne{#1}{(}#1\ifne{#1}{)}\else $R\ifne{#1}{(}#1\ifne{#1}{)}$\fi}
\renewcommand{\b}[1]{\relax\ifmmode G_{#1}\else $G_{#1}$\fi}

\def\@div#1#2{\global\divide \csname c@#1\endcsname#2}
\def\mul#1#2{\global\multiply \csname c@#1\endcsname#2}
\newcounter{Fs}
\newcommand{\Fo}{\ifodd\value{Fs}\relax\else(\fi}
\newcommand{\Fc}{\ifodd\value{Fs}\relax\else)\fi}
\def\berreb{\relax}
\newcommand{\ife}[2]{\ifx\berreb#1\berreb #2 \else \relax \fi}
\newcommand{\ifne}[2]{\ifx\berreb#1\berreb \relax \else #2 \fi}
\newcommand{\F}[2]{\mbox{$
 \Fo
 \mul{Fs}{2}\addtocounter{Fs}{1}
 #1
 \ifne{#2}{\:}
 \@div{Fs}{2}\mul{Fs}{2}
 #2
 \@div{Fs}{2}
 \Fc
$}}
\newcommand{\Bin}[1]{\mbox{$\langle\ife{#1}{\;}#1\rangle$}}
\newcommand{\size}[1]{\mbox{$\# #1$}}
\newcommand{\rank}[1]{\mbox{$\dagger #1$}}
\newcommand{\prank}[1]{\mbox{$\ddagger #1$}}
\newcommand{\Log}[2]{\hj} 
\newcommand{\abs}[1]{\mbox{$\mid\!\!#1\!\!\mid$}}
\newcommand{\sz}[1]{\abs{#1}}
\newcommand{\Bar}{\mbox{\raisebox{0.2ex}{\framebox[0.3em]{\rule{0em}{0.80ex}}}}}
\newcommand{\Neb}[3]{\ifne{#1}{#2}#1\ifne{#1}{#3}}
\newcommand{\A}[2]{\mbox{${\cal A}\Neb{#1}{[}{]}\Neb{#2}{(}{)}$}}
\newcommand{\T}[2]{\mbox{${\cal T}\Neb{#1}{[}{]}\Neb{#2}{(}{)}$}}
\newcommand{\Pot}[1]{\Phi{#1}}
\newcommand{\pot}[2]{\varphi(#1,#2)}
\newcommand{\ch}[2]{\mbox{$#1\raisebox{-.3mm}{$\,\Join\,$}#2$}}   

\newcommand{\node}{\makebox(0,0){$_\bullet$}}

\newsavebox{\BB}
\newsavebox{\BC}
\newsavebox{\BD}
\newsavebox{\BE}
\newsavebox{\BF}
\newsavebox{\BG}
\newsavebox{\BH}
\newsavebox{\BI}
\newsavebox{\BJ}
\newsavebox{\BK}
\newsavebox{\BL}
\newsavebox{\BM}
\newsavebox{\BN}
\newsavebox{\BCT}
\newsavebox{\BET}
\newsavebox{\BIT}
\newsavebox{\BNT}

\newsavebox{\ARC}
\savebox{\ARC}{\begin{picture}(4,0)
 \thicklines
 \put(0.5,0){\vector(1,0){3}}
\end{picture}}
\newcommand{\arc}[1]{\begin{picture}(4,4)
 \put(2,2.5){\makebox(0,0)[b]{#1}}
 \put(0,2){\usebox{\ARC}}
\end{picture}}

\savebox{\BB}{}
\savebox{\BC}{\begin{picture}(0.5,1)(-0.5,-1)
 \put(0,0){\node}
 \put(-0.5,-1){\node}
 \put(0,0){\line(-1,-2){0.5}}
\end{picture}}
\savebox{\BCT}{\begin{picture}(0.5,1)(-0.5,-1)
 \put(0,0){\node}
 \put(0.5,-1){\node}
 \put(0,0){\line(1,-2){0.5}}
\end{picture}}
\savebox{\BD}{\begin{picture}(1,1)(-0.5,-1)
 \put(0,0){\node}
 \put(-0.5,-1){\node}
 \put(0.5,-1){\node}
 \put(0,0){\line(-1,-2){0.5}}
 \put(0,0){\line(1,-2){0.5}}
\end{picture}}
\savebox{\BE}{\begin{picture}(3.34,2)(-1.67,-2)
 \put(0,0){\node}
 \put(-0.67,-1){\makebox(0,0)[tr]{\usebox{\BC}}}
 \put(0.67,-1){\node}
 \put(0,0){\line(-2,-3){0.67}}
 \put(0,0){\line(2,-3){0.67}}
\end{picture}}
\savebox{\BET}{\begin{picture}(3.34,2)(-1.67,-2)
 \put(0,0){\node}
 \put(-0.67,-1){\node}
 \put(0.67,-1){\makebox(0,0)[tr]{\usebox{\BCT}}}
 \put(0,0){\line(-2,-3){0.67}}
 \put(0,0){\line(2,-3){0.67}}
\end{picture}}
\savebox{\BF}{\begin{picture}(3.34,2)(-1.67,-2)
 \put(0,0){\node}
 \put(-0.67,-1){\makebox(0,0)[t]{\usebox{\BD}}}
 \put(0.67,-1){\node}
 \put(0,0){\line(-2,-3){0.67}}
 \put(0,0){\line(2,-3){0.67}}
\end{picture}}
\savebox{\BG}{\begin{picture}(3.34,2)(-1.67,-2)
 \put(0,0){\node}
 \put(-0.67,-1){\makebox(0,0)[t]{\usebox{\BD}}}
 \put(0.67,-1){\makebox(0,0)[tr]{\usebox{\BC}}}
 \put(0,0){\line(-2,-3){0.67}}
 \put(0,0){\line(2,-3){0.67}}
\end{picture}}
\savebox{\BH}{\begin{picture}(4.34,3)(-2.17,-3)
 \put(0,0){\node}
 \put(-1,-1){\makebox(0,0)[t]{\usebox{\BE}}}
 \put(1,-1){\makebox(0,0)[tr]{\usebox{\BC}}}
 \put(0,0){\line(-1,-1){1}}
 \put(0,0){\line(1,-1){1}}
\end{picture}}
\savebox{\BI}{\begin{picture}(4.34,3)(-2.17,-3)
 \put(0,0){\node}
 \put(-1,-1){\makebox(0,0)[t]{\usebox{\BE}}}
 \put(1,-1){\makebox(0,0)[t]{\usebox{\BD}}}
 \put(0,0){\line(-1,-1){1}}
 \put(0,0){\line(1,-1){1}}
\end{picture}}
\savebox{\BIT}{\begin{picture}(4.34,3)(-2.17,-3)
 \put(0,0){\node}
 \put(-1,-1){\makebox(0,0)[t]{\usebox{\BD}}}
 \put(1,-1){\makebox(0,0)[t]{\usebox{\BET}}}
 \put(0,0){\line(-1,-1){1}}
 \put(0,0){\line(1,-1){1}}
\end{picture}}
\savebox{\BJ}{\begin{picture}(4.34,3)(-2.17,-3)
 \put(0,0){\node}
 \put(-1,-1){\makebox(0,0)[t]{\usebox{\BF}}}
 \put(1,-1){\makebox(0,0)[t]{\usebox{\BD}}}
 \put(0,0){\line(-1,-1){1}}
 \put(0,0){\line(1,-1){1}}
\end{picture}}
\savebox{\BK}{\begin{picture}(4.34,3)(-2.17,-3)
 \put(0,0){\node}
 \put(-1,-1){\makebox(0,0)[t]{\usebox{\BG}}}
 \put(1,-1){\makebox(0,0)[t]{\usebox{\BD}}}
 \put(0,0){\line(-1,-1){1}}
 \put(0,0){\line(1,-1){1}}
\end{picture}}
\savebox{\BL}{\begin{picture}(4.34,3)(-2.17,-3)
 \put(0,0){\node}
 \put(-1,-1){\makebox(0,0)[t]{\usebox{\BG}}}
 \put(1,-1){\makebox(0,0)[t]{\usebox{\BE}}}
 \put(0,0){\line(-1,-1){1}}
 \put(0,0){\line(1,-1){1}}
\end{picture}}
\savebox{\BM}{\begin{picture}(6.17,4)(-3.67,-4)
 \put(0,0){\node}
 \put(-1.5,-1){\makebox(0,0)[t]{\usebox{\BH}}}
 \put(1.5,-1){\makebox(0,0)[t]{\usebox{\BE}}}
 \put(0,0){\line(-3,-2){1.5}}
 \put(0,0){\line(3,-2){1.5}}
\end{picture}}
\savebox{\BN}{\begin{picture}(6.17,4)(-3.67,-4)
 \put(0,0){\node}
 \put(-1.5,-1){\makebox(0,0)[t]{\usebox{\BI}}}
 \put(1.5,-1){\makebox(0,0)[t]{\usebox{\BE}}}
 \put(0,0){\line(-3,-2){1.5}}
 \put(0,0){\line(3,-2){1.5}}
\end{picture}}
\savebox{\BNT}{\begin{picture}(6.17,4)(-3.67,-4)
 \put(0,0){\node}
 \put(-1.5,-1){\makebox(0,0)[t]{\usebox{\BE}}}
 \put(1.5,-1){\makebox(0,0)[t]{\usebox{\BIT}}}
 \put(0,0){\line(-3,-2){1.5}}
 \put(0,0){\line(3,-2){1.5}}
\end{picture}}

\newcommand{\ba}{$_{\Bin{}}$}
\newcommand{\bb}{$\node$}
\newcommand{\bc}{\usebox{\BC}}
\newcommand{\bd}{\usebox{\BD}}
\newcommand{\be}{\usebox{\BE}}
\newcommand{\bF}{\usebox{\BF}}
\newcommand{\bg}{\usebox{\BG}}
\newcommand{\bh}{\usebox{\BH}}
\newcommand{\bi}{\usebox{\BI}}
\newcommand{\bj}{\usebox{\BJ}}
\newcommand{\bk}{\usebox{\BK}}
\newcommand{\bl}{\usebox{\BL}}
\newcommand{\bm}{\usebox{\BM}}
\newcommand{\bn}{\usebox{\BN}}
\newcommand{\bnt}{\usebox{\BNT}}

\begin{document}

\maketitle

\begin{abstract}
Leftist heaps and skew heaps are two well-known data structures for
mergeable priority queues. Leftist heaps are constructed for efficiency in
the worst-case sense whereas skew heaps are self-adjusting, designed for
efficiency in the amortized sense. In this paper, we analyze the amortized
complexity of leftist heaps to initiate a full performance comparison with
skew heaps. We consider both the leftist heaps originally developed by Crane
and Knuth, which are also referred to as rank-biased (or, height-biased) leftist heaps,
and the weight-biased leftist heaps introduced by Cho and Sahni.
We show how weight-biased leftist heaps satisfy the same exact
amortized bounds as skew heaps. With these matching bounds we establish a nice
trade-off in which storage of weights is used to limit
the worst-case complexity of leftist heaps, without affecting the amortized
complexity compared to skew heaps. For rank-biased leftist heaps,
we obtain the same amortized lower bounds as for skew heaps, but
whether these bounds are tight is left as an open problem.
\end{abstract}

\section{Celebration}
It's a true pleasure to dedicate this paper to the celebration of Joost-Pieter Katoen's sixtieth birthday.
In 1988 we got to know each other as fresh computer science graduates, sharing office ``Club 737'' at TU Eindhoven
with two more students. I have very fond memories of those days. We were hardworking but also had a lot of fun.
One thing that springs to mind is how Joost (as I got to know him) would throw his pencil
eraser full-force against the blinds, distressing whoever was sitting behind {\em our} Olivetti
desktop PC facing the window. And this was not the only way he kept us sharp. Joost did his
master's at the University of Twente and together with his no-nonsense attitude this was
a welcome change for the rest of us, all from Eindhoven. Our Club 737 participation to the
Nuenen relay triathlon with Pieter Struik swimming, Joost cycling, and me running is
also memorable, not the least for the announcement of our disqualification even before I was finishing
in first position (ever) because Joost had been sent the wrong way, making an early U-turn.

He quickly picked up the ``calculational style'' of programming which was pretty much
baked into every computer science course at TU Eindhoven those days, since E.W.\ Dijkstra started
this in the seventies. We did this not only for imperative programming, using the Guarded Command Language,
but also for procedural, functional, concurrent, parallel, and even logical programming.
One particular topic we happened to work on together was systolic arrays, which later resulted
in two joint papers~\cite{KS91parallel,KS96}. Joost made sure that we got those two papers published,
not just leave it as an internal project.

I think he pretty much continued that way, working on one project after the other, building
this impressive and massive research portfolio we know him for today. Thinking about a topic for
the present paper, I recalled seeing their recent paper on amortized {\em expected} complexity~\cite{BKK+23},
which combines probabilistic reasoning and amortized analysis in a formal setting.
This fits nicely with the work by Tobias Nipkow et al.\ (see, e.g.,~\cite{NB19}) on formal verification
of previously published amortized analyses, in the same style as initiated in my PhD thesis~\cite{Sch92}.
Another interesting development in this direction is the work by Lorenz Leutgeb et al.,
continuing a series of papers by the late Martin Hofmann (see~\cite{LMZ21} and references therein).
In fact, the latter paper presents a slightly improved analysis of splay trees compared to the bounds
in~\cite{ST85}: by changing the potential function into $\Phi'{(x)}=\Pot{(x)}-\log_2\sz{x}$
the amortized costs for the delete operation are reduced from $6\log_2\sz{x}$ to $5\log_2\sz{x}$ comparisons,
while retaining the bounds for the other operations. The new bound was discovered by their automated
analysis of splaying, which builds on elements of type checking systems, starting from the analysis in~\cite[Section 11.2]{Sch92}
and~\cite{Sch93}. In light of the linear constraint systems used in~\cite{LMZ21} we also like
to mention the use of linear programming to find (optimal) potential functions by Sleator~\cite{Sle92}.

With all these developments in mind and with the realization that Joost-Pieter is still going strong
at contributing to these topics in modern computer science, I'll be happy to conclude this paper
with a challenging open problem, conjecturing that the lower bound for rank-biased leftist heaps
presented below is in fact tight.

\section{Leftist Heaps}
Leftist heaps were invented by Crane~\cite{Cra72} and further developed by Knuth, see~\cite[Section~5.2.3]{Knu98}.
Apart from supporting the usual priority queue operations such as $\INS{}{}$ and $\DLM{}$,
leftist heaps also support $\UNI{}{}$ for merging (or, melding) heaps.
Leftist heaps achieve logarithmic time for all these operations. Nowadays, there exist lots
of alternative priority queue implementations, often with better performance, such
as merging in constant time (see, e.g., \cite{Bro13}), but leftist heaps, and skew heaps as their
self-adjusting counterparts, are very simple to implement and generally applicable.

Leftist heaps are slightly more complicated than (top-down) skew heaps. Skew heaps are
self-adjusting and require no storage other than the storage required
for the heap itself. Leftist heaps either store the weight or the rank for each subtree.
Weight-biased leftist heaps were introduced by Cho and Sahni~\cite{CS96} as a single
(top-down) pass variant of the original rank-biased leftist heaps due to Crane and Knuth.
As highlighted by Okasaki~\cite{Oka99}, the weight-biased variant is also favorable in
combination with lazy evaluation; similar advantages apply in combination with concurrent
evaluation, as first considered by Jones~\cite{Jon89}. As skew heaps keep no additional
information at all, the same advantages apply there as well.

\begin{figure}[t]
\begin{center}
\fbox{\begin{tabular}{l}
$\begin{array}{lcl}
 \EMP &=& \Bin{} \\
 \ISE{x} &=& x=\Bin{} \\
 \SIN{a} &=& \Bin{a} \\
 \UNI{x}{y} &=& \ch{x}{y} \\
 \MNM{\Bin{}} &=& \infty \\
 \MNM{\Bin{t,a,u}} &=& a \\
 \DLM{\Bin{t,a,u}} &=& \ch{t}{u}
 \end{array}$
 \vspace*{1mm} \\
 $\begin{array}{lcll}
 \ch{\Bin{}}{\Bin{}} &=& \Bin{}  \\
 \ch{\Bin{t,a,u}}{y} &=& \bal{\Bin{t,a,\ch{u}{y}}}, &a\leq\MNM{y} \\
 \ch{x}{\Bin{t,a,u}} &=& \bal{\Bin{t,a,\ch{x}{u}}}, &\mbox{otherwise}
\end{array}$
\end{tabular}}

\vspace*{3mm}
\fbox{$\begin{array}{c@{\hspace{3mm}}lcll}
 \mbox{Skew heaps} & \bal{\Bin{t,a,u}} &=& \Bin{u,a,t} \\[1mm]
 \mbox{Weight-biased} &
 \bal{\Bin{t,a,u}}  &=& \Bin{t,a,u}, &\size{t}>\size{u}\\
 \mbox{leftist heaps} &&\Bar& \Bin{u,a,t}, &\size{t}\leq\size{u}\\[1mm]
 \mbox{Rank-biased} &
 \bal{\Bin{t,a,u}}  &=& \Bin{t,a,u}, &\rank{t}>\rank{u} \\
 \mbox{leftist heaps} &&\Bar& \Bin{u,a,t}, &\rank{t}\leq\rank{u}
\end{array}$}
\end{center}
\caption{Purely functional binary heaps with three balancing strategies.}
\label{tdsh}
\end{figure}
A simple implementation of the three types of heaps is shown in Figure~\ref{tdsh},
omitting the redundant operation $\INS{a}{x}=\UNI{\SIN{a}}{x}$.
We use \Bin{} to denote an empty tree and \Bin{t,a,u} to denote a binary tree with
left subtree $t$, root value $a$, and right subtree $u$, with \Bin{a} as a shorthand
for \Bin{\Bin{},a,\Bin{}}. The weight \size{x}
of tree $x$ is defined by $\size{\Bin{}}=0$ and $\size{\Bin{t,a,u}}=\size{t}+\size{u}+1$.
The rank \rank{x} is defined by $\rank{\Bin{}}=0$ and $\rank{\Bin{t,a,u}}=\rank{u}+1$,
hence \rank{x} is equal to the length of the rightmost path from the root of $x$.
Binary heaps satisfy the heap property implying that the root holds the minimum value
for each subtree. In addition,  for weight-biased heaps, the weight of each left subtree
is at least the weight of its right sibling, and for rank-biased heaps,
the rank of each left subtree is at least the rank of its right sibling.
Note that weights require about $\log_2 n$ bits of storage for size-$n$ subtrees,
while $\log_2 \log_2 n$ bits suffice for ranks, due to (\ref{eq:rank}).

As a suitable cost measure $\T{}{}$ for the time complexity of binary heaps we count the number of
comparisons $a\leq\MNM{y}$ used for unfolding $\ch{}{}$, including all comparisons with $\infty$.
Clearly, the costs for $\EMP{}$, $\ISE{x}$, $\SIN{a}$, $\MNM{x}$ are then 0,
while $\T{\UNI{}{}}{x,y} = \rank{x}+\rank{y}$  and $\T{\DLM{}}{\Bin{t,a,u}} = \rank{t} + \rank{u}$.
For leftist heaps we have the bound
\begin{equation}\label{eq:rank}
\rank{x} \leq \log_2 \sz{x},
\end{equation}
using $\sz{x}=\size{x}+1$ to denote the weight of $x$ plus one~\cite{Cra72,Knu98,CS96}.
This leads to the following worst-case bounds for leftist heaps:
\[\begin{array}{lll}
\T{\UNI{}{}}{x,y} &\leq& 2\,\log_2 \sz{\ch{x}{y}}, \\
\T{\DLM{}}{x}     &\leq& 2\,\log_2 \sz{x}.
\end{array}\]

For skew heaps, \rank{x} is linear in $\sz{x}$ in the worst case, so we need
to perform an amortized analysis to get any useful bounds.
However, as we will show below, the above worst-case bounds are also overly pessimistic for leftists heaps.
We will present several improved amortized analyses for leftist heaps.

Note that by $\T{f}{E}$ we denote the cost of evaluating $f$ given the {\em value} of expression $E$,
not counting the cost of the evaluation of $E$~\cite{KS91}:
\[ \T{f}{E}=\T{}{f(E)}-\T{}{E}.\]
This allows for modular analysis using composition rules such as
\[ \T{g \circ f}{E} = \T{f}{E} + \T{g}{f(E)}.\]
Defining the amortized cost of $f$ with respect to a potential function $\Pot{}$ as
\[ \A{f}{E} = \T{f}{E} + \Pot{(f(E))} - \Pot{(E)} \]
then allows for modular amortized analysis, using composition rules such as
 \[ \A{g \circ f}{E} = \A{f}{E} + \A{g}{f(E)}. \]
See \cite[Ch.~5]{Sch92} for this formalization of the physicist's method~\cite{ST86,Tar85}.

\section{Simple Amortized Analysis of Leftist Heaps}
\label{section:simple-bounds}
We perform an amortized analysis in terms of a potential function $\Pot{}$,
where the amortized costs are given by:
\[\begin{array}{lllll}
\A{\UNI{}{}}{x,y} &=& \T{\ch{}{}}{x,y} + \Pot{(\ch{x}{y})} -\Pot{(x)} -\Pot{(y)}, \\
\A{\DLM{}{}}{x} &=& \T{\ch{}{}}{t,u} + \Pot{(\ch{t}{u})}  - \Pot{(x)}, \mbox{\ \  with } x=\Bin{t,a,u}.
\end{array}\]
Taking $\Pot{(x)}=\rank{x}$ gives us the following bounds for leftist heaps, using (\ref{eq:rank}):
\[\begin{array}{lllll}
\A{\UNI{}{}}{x,y} &=& \rank{x}+\rank{y} + \rank{(\ch{x}{y})} - \rank{x} -\rank{y}  &\leq&  \log_2 \sz{\ch{x}{y}}, \\
\A{\DLM{}{}}{x} &=& \rank{t} + \rank{u} + \rank{(\ch{t}{u})} - \rank{u} - 1 &\leq&  2\,\log_2 \sz{x}.
\end{array}\]
This means that the bound for \UNI{}{} is reduced by 50\%, while retaining the bound for \DLM{}.
Note that (\ref{eq:rank}) actually implies the same bounds for any potential function $\Pot{}$
satisfying $\rank{x}\leq\Pot{(x)}\leq\log_2 \sz{x}$.

\section{Amortized Analysis of Weight-Biased Leftist Heaps}
\label{section:amor-weight}
We have reduced the bound for \UNI{}{} by 50\% but for many applications the overall cost is dominated by the cost of the \DLM{} operations.
Therefore we want to make the bound for \DLM{} as tight as possible. In this section we will show that for weight-biased leftist heaps
the upper bound for \A{\DLM{}}{x} can be improved to $\log_\phi \sz{x}\approx1.44\,\log_2 \sz{x}$, where $\phi=(\sqrt{5}{+}1)/2\approx1.618$ is the golden ratio.
The lower bound in Section~\ref{section:lower-bound} will imply that this bound is tight.

In our amortized analysis we exploit the close relationship between weight-biased leftist heaps and top-down skew heaps.
We define potential function $\Pot{}$ as
\[\begin{array}{lll}
\Pot{\Bin{}}      &=& 0, \\
\Pot{\Bin{t,a,u}} &=& \Pot{(t)} + \pot{t}{u} + \Pot{(u)},
\end{array}\]
where
\[ \pot{t}{u} = \max\left(\log_\beta \frac{\beta\sz{u}}{\sz{t}+\sz{u}}, 0\right), \]
with $\beta=\phi^{\phi+2}\approx 5.703$, following~\cite{KS91},~\cite[Lemma~9.2]{Sch92}.

The essential property that we will need for this potential function is that
\begin{equation}\label{eq:phi-prop}
\sz{t}\leq\sz{u} \Rightarrow \pot{u}{t}\leq\pot{t}{u},
\end{equation}
as $\log_\beta$ is increasing. Also note that $0\leq\pot{t}{u}\leq1$.

The central operation to analyze is \ch{x}{y}. Because of the balancing step we cannot reuse the analysis of~\cite{KS91}.
Instead we will directly prove by induction on $x$ and $y$ that
\begin{equation}\label{eq:indhyp}
\A{\ch{}{}}{x,y} \leq \log_\alpha \sz{\ch{x}{y}} + \log_\beta \sz{x} + \log_\beta \sz{y},
\end{equation}
with $\alpha=\phi^{2\phi-1}\approx 2.933$.
This is clearly true for $x=y=\Bin{}$. Otherwise, assume w.l.o.g.\ that $x=\Bin{t,a,u}$ and $a\leq\MNM{y}$.
Then $\ch{x}{y}=\bal{\Bin{t,a,\ch{u}{y}}}$ and
\[\begin{array}{lll}
 \A{\ch{}{}}{x,y}
 &=& \T{\ch{}{}}{x,y} + \Pot{(\ch{x}{y})} - \Pot{(x)} - \Pot{(y)} \\[1mm]
 &=& 1 + \T{\ch{}{}}{u,y} + \Pot{(\rm bal \Bin{t,a,\ch{u}{y}})} -\Pot{\Bin{t,a,u}} - \Pot{(y)} \\[1mm]
 &=& \A{\ch{}{}}{u,y} + 1 - \pot{t}{u} + \left\{\begin{array}{ll}
    \pot{t}{\ch{u}{y}} , & \mbox{if } \sz{t}>\sz{\ch{u}{y}} \\
    \pot{\ch{u}{y}}{t} , & \mbox{if } \sz{t}\leq\sz{\ch{u}{y}}
    \end{array}\right. \\[1mm]
 &\leq& \A{\ch{}{}}{u,y} + 1 - \pot{t}{u} + \pot{\ch{u}{y}}{t} \\[1mm]
 &\leq& \log_\alpha \sz{\ch{u}{y}} +\pot{\ch{u}{y}}{t} + \log_\beta \sz{u} + 1 - \pot{t}{u} + \log_\beta \sz{y} \\[1mm]
 &\leq& \log_\alpha \sz{\ch{u}{y}} +\pot{\ch{u}{y}}{t} + \log_\beta (\sz{t} +\sz{u}) + \log_\beta \sz{y} \\[1mm]
 &\leq& \log_\alpha (\sz{t} + \sz{\ch{u}{y}}) + \log_\beta \sz{x} + \log_\beta \sz{y} \\[1mm]
 &=& \log_\alpha \sz{\ch{x}{y}} + \log_\beta \sz{x} + \log_\beta \sz{y},
\end{array}\]
where the first inequality follows from (\ref{eq:phi-prop}),
the second inequality follows from the induction hypothesis,
the third inequality follows from the definition of $\varphi$,
and the last inequality follows as it was shown in~\cite{KS91} that
\[ \log_\beta \frac{\beta n}{m+n} \leq \log_\alpha \frac{m+n}{m}, \]
for all positive $m$ and $n$.

With both \sz{x} and \sz{y} bounded by \sz{\ch{x}{y}}, we conclude from (\ref{eq:indhyp}) that
\[ \A{\ch{}{}}{x,y} \leq \log_\phi \sz{\ch{x}{y}}, \]
as $\log_\alpha \phi + 2\log_\beta \phi = 1/(2\phi-1) + 2/(\phi+2) = 1$ (using $\phi^2=\phi+1$).
With this we have proved our first main result.
\begin{theorem}{}\label{ubs} There exists a potential function
such that the amortized costs for weight-biased leftist heaps satisfy (counting comparisons): \EMP, \ISE{x}, \SIN{a},
and \MNM{x} cost 0, \UNI{x}{y} costs at most $\log_\phi (\sz{x}+\sz{y})$ and \DLM{x} costs at most $\log_\phi \sz{x}$.
\end{theorem}
These bounds coincide with the bounds for skew heaps from~\cite{KS91}. Intuitively, one may say that for weight-biased heaps,
not swapping in case the left subtree is already larger does not hurt the performance compared to skew heaps, where we always
swap subtrees.

The above bounds may now be combined with the bounds from Section~\ref{section:simple-bounds} by taking their convex combination:
\[\begin{array}{lllll}
\A{\UNI{}{}}{x,y} &\leq& \lambda\,\log_2 \sz{\ch{x}{y}} + (1-\lambda)\log_\phi \sz{\ch{x}{y}} &\approx& (1.44-.44\lambda)\log_2 \sz{\ch{x}{y}}, \\
\A{\DLM{}}{x} &\leq& 2\lambda\,\log_2 \sz{x} +(1-\lambda)\log_\phi \sz{x} &\approx& (1.44+.56\lambda)\log_2 \sz{x},
\end{array}\]
with $0\leq\lambda\leq1$.

\section{Amortized Analysis of Rank-Biased Leftist Heaps}
\label{section:amor-rank}
For rank-biased heaps we do not have a similar improvement over the bounds from Section~\ref{section:simple-bounds}.
In a rank-biased heap there is no relation between the weights of left and right subtrees,
and we cannot use property (\ref{eq:phi-prop}) like we can for weight-biased heaps.
In fact, as we will prove below, the potential functions $\Pot{}$ from the previous sections do
not work for rank-biased leftists heaps. The search for a potential function $\Pot{}$ for which
$\A{\DLM{}}{x}\leq c\,\log_2 \sz{x}$ for some constant $c<2$ is therefore completely open.\footnote{
Of course, dismissing potential functions such as $\Pot{(x)}=\gamma\,\size{x}\,\log_2\sz{x}$, $\gamma>0$,
which artificially shift costs from \DLM{} to \UNI{}{}: $\A{\DLM{}}{x}\leq(2-\gamma)\log_2\sz{x}$
can be made arbitrarily small, but $\A{\UNI{}{}}{x,y}$ gets ineffectively large, e.g., linear in $\sz{x}$ and $\sz{y}$.}

Let $B_k^a$ denote a perfect binary tree of height $k$ in which all nodes are labeled with $a$,
that is, $B_0^a=\Bin{}$ and $B_k^a=\Bin{B_{k-1}^a,a,B_{k-1}^a}$ for $k\geq1$.
We introduce trees $W_k=\Bin{T_k,0,U_k}$ for $k\geq2$, where
\[\begin{array}{l@{\hspace{1cm}}l}
T_2 = \Bin{\Bin{0},0,\Bin{2}}, & U_2 = \Bin{B_2^1,1,\Bin{1}}, \\[0.6mm]
T_k = \Bin{B_k^0,0,T_{k-1}},   & U_k = \Bin{B_k^1,1,U_{k-1}}.
\end{array}\]
As these trees satisfy the heap property and $\rank{B_k}=\rank{T_k}=\rank{U_k}=k$,
it follows that $W_k$ is indeed a rank-biased leftist heap for all $k\geq2$.
Also, $\sz{B_k}=2^k$, $\sz{T_k}=2^{k+1}-4$, and $\sz{U_k}=2^{k+1} -2$ (hence $W_k$ is
also weight-biased).

We will now analyze the amortized cost of $\DLM{W_k}=V_k$, with $V_k=\ch{T_k}{U_k}$.
We first do so for potential function $\Pot{(x)}=\rank{x}$ from Section~\ref{section:simple-bounds}.
Without proof we state that $\rank{V_k}=k$ as well, so we get
\[\A{\DLM{}}{W_k} = \rank{T_k} + \rank{U_k} +\rank{V_k} - \rank{W_k} = 2k-1.\]
Since $\sz{W_k}=2^{k+2}-6$, we see that $\A{\DLM{}}{W_k} \approx 2 \log_2 \sz{W_k}$.
Thus, for trees $W_k$, we get the worst-case complexity for $\DLM{}$, taking $2k$ comparisons
(of which only the last one is a comparison with $\infty$), while at the same time the potential difference
$\Pot{(V_k)} - \Pot{(W_k)}=-1$ is close to zero.

Together with the fact that trees $W_k$ can actually
occur as rank-biased heaps (see Appendix~\ref{section:surjectivity}), we clearly cannot get
an improved bound for the amortized cost of $\DLM{}$ when using $\Pot{(x)}=\rank{x}$.
The same conclusion can be drawn for any potential function satisfying $\rank{x}\leq\Pot{(x)}\leq\log_2\sz{x}$.
This conclusion even extends to functions like $\Pot{(x)}=\prank{x}$. Here, the ``minor'' rank $\prank{x}$ is defined
by $\prank{\Bin{}}=0$ and $\prank{\Bin{t,a,u}}=\rank{t}+1$, for which we state the following
properties without proof.
\begin{lemma}{}\label{lem:rankmeld}
For rank-biased heaps $x$ and $y$:
\begin{enumerate}
\item [\rm (i)] $\rank{x}\leq\prank{x}\leq \log_2 \sz{x} + 1$,
\item[\rm (ii)] $\min(\rank{x},\rank{y}) \leq \rank{(\ch{x}{y})} \leq \rank{x}+\rank{y}$,
\item[\rm (iii)] $\rank{(\ch{x}{y})} \leq \max(\prank{x},\prank{y})$,
\item[\rm (iv)] $\min(\prank{x},\prank{y}) \leq \prank{(\ch{x}{y})} \leq \prank{x}+\prank{y}$.
\end{enumerate}
\end{lemma}
Using $\Pot{(x)}=\prank{x}$, property (i) yields essentially the same bounds for \A{\UNI{}{}}{}
and \A{\DLM{}}{} as in Section~\ref{section:simple-bounds}.

Finally, we consider the potential function $\Pot{}$ from Section~\ref{section:amor-weight},
including all variants of $\pot{t}{u}$ proposed in~\cite{KS91,Sch92}. Typically, $\pot{t}{u}$ is
decreasing in $\sz{t}$ and increasing in $\sz{u}$, and possibly $0\leq\pot{t}{u}\leq1$.
As it turns out, the combination of these properties does not work out well for $\A{\DLM{}}{W_k}$,
as by construction the trees $V_k$ are {\em not} weight-biased.
We basically get the same results for the amortized cost of $\DLM{W_k}$ as
for the potential functions from Section~\ref{section:simple-bounds}.
Concretely, e.g., for $\pot{t}{u}= \log_\beta (\beta\sz{u}/(\sz{t}+\sz{u}))$, we have, deferring the proof
to Appendix~\ref{section:proof4}:
\begin{equation}\label{eq:potdiff}
 \Pot{(V_k)} - \Pot{(W_k)} = -1 + \log_\beta \frac{(2^{k+2} -6)(2^{k+1} -4)}{(2^{k+2} -7)(2^{k+1} -1)}.
\end{equation}
We have already seen that $\T{\DLM{}}{W_k}=2k$, so $\A{\DLM{}}{W_k} \approx 2k-1$,
and we reach the same conclusion as above.

In the next section, however, we will see that the same lower bound applies to rank-biased heaps as for weight-biased heaps.
This leaves a gap between the lower bound $\log_\phi \sz{x}$ and the upper bound  $2\,\log_2 \sz{x}$ for \A{\DLM{}}{x}.
Resolving this gap for rank-biased leftist heaps therefore remains as an open problem.

\section{Lower Bound for Leftist Heaps}
\label{section:lower-bound}
\begin{figure}[t]
\begin{center}
\begin{tabular}{ccccccccc}
\ba&\bb&\bc&\bd&\be&\bF&\bg&\bh&\bi \\
\b{0}&\b{1}&\b{2}&\b{3}&\b{4}&\b{5}&\b{6}&\b{7}&\b{8}\\[5mm]
\end{tabular}
\begin{tabular}{c@{\hspace{-2mm}}c@{\hspace{-2mm}}c@{\hspace{-2mm}}c@{\hspace{-2mm}}c}
\bj&\bk&\bl&\bm&\bn \\
\b{9}&\b{10}&\b{11}&\b{12}&\b{13}
\end{tabular}
\end{center}
\caption{\b{n} trees for $n=0,\ldots,13$ \cite[Figure~2]{Sch97tight}.}
\label{fig}
\end{figure}
In this section we show that the worst-case sequences of operations on skew heaps from \cite{Sch97tight} also apply
to weight-biased leftist heaps and rank-biased leftist heaps. We will first recall the golden
trees $G_n$ from \cite{Sch97tight}, and then show that these trees are weight-biased and also rank-biased.

We let $\l{n}$ denote Hofstadter's G-sequence \cite{Hof79} defined by $\l{0}=0$ and $\l{n}=n-\l{\l{n-1}}$, $n\geq1$, see also \url{www.oeis.org/A005206}.
Further, let $\r{n}=n-\l{n}$. A golden tree $\b{n}$ of order $n$, $n\geq0$, is defined recursively as the following
unlabeled binary tree:
\[\begin{array}{lcll}
 \b{0} &=& \Bin{}, \\
 \b{n} &=& \Bin{\b{\l{n-1}},\b{\r{n-1}}}, & n\geq1.
\end{array}\]
The golden trees form a supersequence of the (unlabeled) Fibonacci trees~\cite[Section~6.2.1]{Knu98} in the sense that
a golden tree $\b{n}$ corresponds to the Fibonacci tree of order $k$ whenever $n=\f{k+2}-1$, $k\geq0$. See Figure~\ref{fig},
trees $\b{0}, \b{1}, \b{2}, \b{4}, \b{7}$, and $\b{12}$ correspond to Fibonacci trees.

Clearly, golden trees are weight-biased leftist trees as $\size{G_n}=n$ and $\l{n}\geq\r{n}$ \cite[Lemma~3]{Sch97tight}.
In fact, $\l{n}/\r{n}$ approaches the golden ratio $\phi$ for larger $n$, which explains the name for these trees.
Since $R(n)$ is nondecreasing \cite[Lemma~3]{Sch97tight}, it is easily seen that $\rank{\b{n}}$ is also nondecreasing in $n$,
which implies that golden trees are also rank-biased leftist trees.

As a consequence, the heaps arising in the worst-case sequences of operations constructed in~\cite[Section 3]{Sch97tight}
will remain exactly the same for weight-biased heaps and for rank-biased heaps. That is,
all applications of \ch{}{} will behave the same for the three types of heaps.
Concretely, in the unlabeled case, this behavior is captured by the following version of \ch{}{}:
\[\begin{array}{lcl}
 \ch{\Bin{}}{\Bin{}} &=& \Bin{}, \\
 \ch{\Bin{t,u}}{y} &=&  \Bin{\ch{y}{u},t}.
\end{array}\]
Each application of \ch{}{} combines two golden trees of appropriate sizes
resulting in another golden tree, see~\cite[Lemma 7]{Sch97tight}:
\[ \b{n} = \ch{\b{\l{n}}}{\b{\r{n}}}.\]
This ``self-recreating'' property is illustrated in Figure~\ref{fig2}, where operation \ch{x}{y} is split into two stages.
First, the rightmost paths of $x$ and $y$ are merged, and then all left and right subtrees of
the nodes on the merged path are swapped. The balancing strategies defined in~Figure~\ref{tdsh},
which break ties by always swapping trees of equal weight or rank, ensure this behavior for the labeled case as well.
\begin{figure}[t]
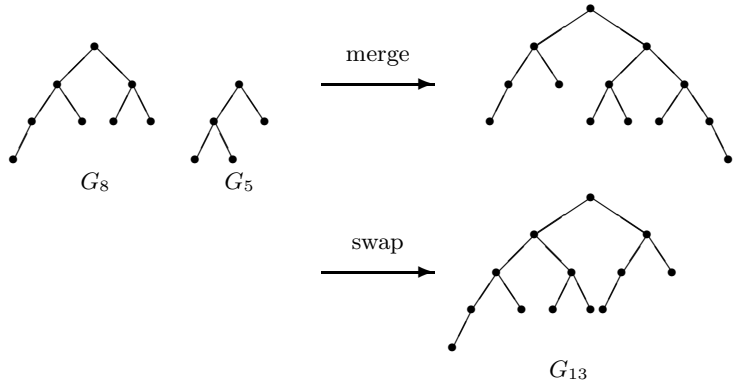

\setlength{\tabcolsep}{0pt}
\begin{center}
\begin{tabular}{cccc}
\bi & \bF & \arc{merge} & \bnt \\
\b{8} & \b{5} \\
& & \arc{swap} & \bn \\
& & & \b{13}
\end{tabular}
\end{center}
\caption{Operation \protect{\ch{\b{8}}{\b{5}}}, viewed as first merging the rightmost paths
and then swapping the subtrees of all nodes on the rightmost path, resulting
in \b{13} \cite[Figure~3]{Sch97tight}.}
\label{fig2}
\end{figure}

With this we have proved our second main result.
\begin{theorem}{} There exist sequences of operations on leftist heaps such that
the heaps may get arbitrarily large and for which the actual costs satisfy (counting comparisons):
\EMP, \ISE{x}, \MNM{x}, and \SIN{a} cost 0, \DLM{x} costs at least $\log_\phi \sz{x}-c$,
and \UNI{x}{y} costs at least $\log_\phi (\sz{x}+\sz{y})-c$, for some constant $c<2$.
\end{theorem}

\section{Concluding Remarks}
We have analyzed the amortized complexity of leftist heaps. Compared to the well-known worst-case bounds,
our amortized bounds provide a more refined view of the computational complexity of leftist heaps.
Another example of such a refined amortized analysis is the analysis of insertions into AVL-trees by
Mehlhorn and Tsakalidis~\cite{MT86}. AVL-trees are height-balanced and achieve logarithmic time for the
main operations in the worst case. The amortized bounds yield a more profound understanding of the
time complexity and may also help in experimental studies of these data structures, e.g., to understand
why skew heaps and leftist heaps behave more similarly than one might anticipate.

For weight-biased heaps, we have been able to reduce the bound for \DLM{x}
from $2\,\log_2 \sz{x}$ to $\log_\phi \sz{x}\approx 1.44\,\log_2 \sz{x}$.
We would like to point out that the improved bounds from~\cite{KS91} are essential for this result.
For the potential function $\Pot{}$ from~\cite{ST86}, with $\pot{t}{u}=1$ if $\sz{t} < \sz{u}$ and $\pot{t}{u}=0$ otherwise, we do not get
an improvement over the worst-case bounds. Indeed, for any weight-biased heap $x$ we would simply have $\Pot{(x)}=0$.

For rank-biased heaps, all potential functions considered in this paper (including the above potential function from~\cite{ST86})
fall short of achieving any improvement over the worst-case bound for \DLM{x}. As explained in Section~\ref{section:amor-rank},
this leaves us with an interesting open problem regarding the computational complexity of rank-biased leftist heaps:
whether the amortized cost of \DLM{x} is equal to $\log_\phi \sz{x}$, equal to $2\,\log_2 \sz{x}$,
or something in between. Hopefully, innovations such as ATLAS~\cite{LMZ21} will help us in cracking this problem,
expediting the search for potential functions beating the worst-case bound for \DLM{}.

Moreover, as another interesting direction for further research (also in light of~\cite{BKK+23}),
we introduce the {\em randomized} leftist heaps.
This probabilistic variant is defined as in Figure~\ref{tdsh}, but with the following balancing strategy:
\[\begin{array}{lcll}
\bal{\Bin{t,a,u}}  &=& \Bin{t,a,u}, &\ \ \mbox{with probability $1-p$}, \\
                &\Bar& \Bin{u,a,t}, &\ \ \mbox{with probability $p$},
\end{array}\]
where $0\leq p\leq1$ is a fixed bias. For $p=1/2$, the behavior of these heaps resembles the randomized meldable heaps of~\cite{GM98},
for which the expected costs of \UNI{x}{y} and \DLM{x} are bounded by $2\log_2\sz{\ch{x}{y}}$ and $2\log_2\sz{x}$, respectively.
Note that these bounds coincide with the worst-case bounds for leftists heaps. For $p=1$, however, the randomized leftist heaps
reduce to skew heaps, for which we have superior amortized bounds.
Thus, the question is which value for $p$ will be optimal in which sense, e.g., minimizing the expected amortized costs {\em and} ensuring
that the worst-case actual costs are also limited with high probability (e.g., what about $p=1/\phi$).
The results for randomized splay trees (see~\cite{LMZ22} and references therein) seem a good starting point for such an amortized analysis.

\medskip

\noindent {\bf Acknowledgments.} Tom Verhoeff and the anonymous reviewers are gratefully acknowledged for their useful comments.

\newcommand{\etalchar}[1]{$^{#1}$}

\appendix
\section{Reachability of Rank-Biased Leftist Heaps}
\label{section:surjectivity}
When we consider trees $W_k$ to argue that $\A{\DLM{}}{W_k}\approx 2\log_2 \sz{W_k}$ for a given potential function $\Pot{}$
in Section~\ref{section:amor-rank}, we have to ensure that these trees can actually be reached using the operations
available for rank-biased leftist heaps.

As shown in Theorem~\ref{thm:anyheap} below, we only need operations $\EMP$, $\SIN{}$, $\UNI{}{}$
from Figure~\ref{tdsh} to reach any given rank-biased leftist heap.
\begin{lemma}{}\label{lem:anyheap}
For rank-biased heap $x$, a rank-biased heap $y$ exists s.t. $\ch{\Bin{}}{y}=x$.
\end{lemma}
\begin{proof} By induction on $x$. If $x=\Bin{}$, we put $y=\Bin{}$. Clearly, $\ch{\Bin{}}{y}=\ch{\Bin{}}{\Bin{}}=\Bin{}=x$.
Otherwise $x=\Bin{t,a,u}$ with $\rank{t}\geq\rank{u}$, and we distinguish two cases.

If $\rank{t}=\rank{u}$, we put $y=\Bin{u,a,z}$, where $z$ is obtained from the
induction hypothesis for $t$, hence $\ch{\Bin{}}{z}=t$, and
\[ \ch{\Bin{}}{y} = \ch{\Bin{}}{\Bin{u,a,z}} = \bal{\Bin{u,a,\ch{\Bin{}}{z}}} = \bal{\Bin{u,a,t}} = \Bin{t,a,u}=x. \]

If $\rank{t}>\rank{u}$, we put $y=\Bin{t,a,z}$, where $z$ is obtained from the
induction hypothesis for $u$, hence $\ch{\Bin{}}{z}=u$, and
\[ \ch{\Bin{}}{y} = \ch{\Bin{}}{\Bin{t,a,z}} = \bal{\Bin{t,a,\ch{\Bin{}}{z}}} = \bal{\Bin{t,a,u}} = \Bin{t,a,u}=x. \]
\end{proof}

\begin{theorem}{}\label{thm:anyheap}
Any rank-biased heap can be generated using $\EMP$, $\SIN{}$, $\UNI{}{}$.
\end{theorem}
\begin{proof}
By induction on $x$, we show that any rank-biased heap $x$ can be generated using operations $\EMP$, $\SIN{}$, and $\UNI{}{}=\ch{}{}$.

If $x=\Bin{}$, then we obtain $x$ as $x=\EMP$.
Otherwise $x=\Bin{t,a,u}$ with $\rank{t}\geq\rank{u}$, and we distinguish three cases.

If $\rank{t}=\rank{u}=0$, then $t=u=\Bin{}$ and we obtain $x=\Bin{a}$ as $x=\SIN{a}$.

If $\rank{t}=\rank{u}\geq1$, we generate $\Bin{u, a, \Bin{}}$ and $y$ using the induction hypothesis twice,
where $y$ is a rank-biased heap such that $\ch{\Bin{}}{y}=t$ (cf.~Lemma~\ref{lem:anyheap}). Then we have
\[ \ch{\Bin{u,a,\Bin{}}}{y} = \bal{\Bin{u,a,\ch{\Bin{}}{y}}} = \bal{\Bin{u,a,t}} = \Bin{t,a,u}=x. \]

If $\rank{t}>\rank{u}$, we generate $\Bin{t, a, \Bin{}}$ and $y$ using the induction hypothesis twice,
where $y$ is a rank-biased heap such that $\ch{\Bin{}}{y}=u$ (cf.~Lemma~\ref{lem:anyheap}). Then we have
\[ \ch{\Bin{t,a,\Bin{}}}{y} = \bal{\Bin{t,a,\ch{\Bin{}}{y}}} = \bal{\Bin{t,a,u}} = \Bin{t,a,u}=x. \]
\end{proof}
The same kind of reachability can be proved for skew heaps and for weight-biased heaps.
See the notion of {\em surjectivity} for monoalgebras~\cite[Def.~2.19]{Sch92} for a formalization of reachability.

\section{Proof of Eq.~(\ref{eq:potdiff})}
\label{section:proof4}
To evaluate the potential difference $\Pot{(V_k)} - \Pot{(W_k)}$,
we will use the skewed view of binary trees from~\cite[Section~6.2.2]{Sch92}, defined as the smallest set $X$ satisfying
\[ X = [X \times A]. \]
Here, an empty list $[\ ]$ represents an empty tree, and a nonempty list $[\Bin{t,a}] +\!\!\!\!\!+ u$ represents
a binary tree with left subtree $t$, root value $a$, and right subtree $u$. We write $[a]$ as a shorthand for $[\Bin{[\ ], a}]$.

The skewed view gives direct access to the rightmost path from the root of a binary tree.
We use this to access the nodes near the end of the rightmost paths in $T_k$, $U_k$, and $V_k=\ch{T_k}{U_k}$:
\begin{eqnarray*}
T_k &:& \quad T'_k +\!\!\!\!\!+ [\Bin{[0], 0}] +\!\!\!\!\!+ [2] \\
U_k &:& \quad U'_k +\!\!\!\!\!+ [1] \\
V_k &:& \quad T'_k +\!\!\!\!\!+ [\Bin{U'_k +\!\!\!\!\!+ [\Bin{[2], 1}], 0}] +\!\!\!\!\!+ [0],
\end{eqnarray*}
where subtree $U'_k +\!\!\!\!\!+ [\Bin{[2], 1}]$ of $V_k$ corresponds to $\ch{\Bin{2}}{U_k}$.

Recalling that $\pot{t}{u}= f(\sz{t},\sz{u})$ with $f(m,n)= \log_\beta (\beta n/(m+n))$,
we evaluate $\Pot{(V_k)} - \Pot{(W_k)}$ as follows:
\[\begin{array}{cl}
&\Pot{(V_k)} - \Pot{(W_k)} \\[2mm]
=&\Pot{(\ch{T_k}{U_k})} - \Pot{\Bin{T_k,0,U_k}} \\[2mm]
=&\Pot{(\ch{T_k}{U_k})} - \Pot{(T_k)} -\pot{T_k}{U_k} -\Pot{(U_k)} \\[2mm]
=&-\pot{T_k}{U_k}\\[1mm]
&+\ \sum_{i=3}^k \left(\pot{B^0_i}{\ch{T_{i-1}}{U_k}}-\pot{B^0_i}{T_{i-1}}\right) \hfill \mbox{``along $T'_k$''}\\[1mm]
&+\ \pot{\ch{\Bin{2}}{U_k}}{\Bin{0}} - \pot{\Bin{0}}{\Bin{2}} \hfill \mbox{``along $[\Bin{[0], 0}] +\!\!\!\!\!+ [2] $''}\\[1mm]
&+\ \sum_{i=2}^k \left(\pot{B^1_i}{\ch{\Bin{2}}{U_{i-1}}}-\pot{B^1_i}{U_{i-1}}\right) \hfill \mbox{``along $U'_k$''}\\[1mm]
&+\ \pot{\Bin{2}}{\Bin{}} - \pot{\Bin{}}{\Bin{}} \hfill \mbox{``along $[1]$''}\\[2mm]
=& -f(2^{k+1}-4,2^{k+1}-2) \\[1mm]
&+\ \sum_{i=3}^k \left(f(2^i,2^i-4+2^{k+1}-2-1)-f(2^i,2^i-4)\right) \\[1mm]
&+\ f(2^{k+1}-1,2)-f(2,2) \\[1mm]
&+\ \sum_{i=2}^k \left(f(2^i,2^i-1)-f(2^i,2^i-2)\right)\\[1mm]
&+\ f(2,1)  -f(1,1) \\[2mm]

=&-\log_\beta \frac{\beta(2^{k+1}-2)}{2^{k+2}-6}+\sum_{i=3}^k \log_\beta \left(\frac{2^i+2^{k+1}-7}{2^i-4} \ \frac{2^i+2^i-4}{2^i+2^i+2^{k+1}-7}\right) \\
&+\ \log_\beta \frac{2+2}{2^{k+1}-1+2}+\sum_{i=2}^k \log_\beta \left(\frac{2^i-1}{2^i-2} \ \frac{2^i+2^i-2}{2^i+2^i-1}\right)+\log_\beta \frac{1+1}{2+1} \\[2mm]

=&-1 - \log_\beta \frac{2^{k+1}-2}{2^{k+2}-6}+\sum_{i=3}^k \log_\beta \left(\frac{2^{i+1}-4}{2^i-4} \ \frac{2^i+2^{k+1}-7}{2^{i+1}+2^{k+1}-7}\right) \\
&+\ \log_\beta \frac{4}{2^{k+1}+1}+\sum_{i=2}^k \log_\beta \left(\frac{2^{i+1}-2}{2^i-2} \ \frac{2^i-1}{2^{i+1}-1}\right)+\log_\beta \frac{2}{3} \\[2mm]

=&-1 + \log_\beta \frac{2^{k+2}-6}{2^{k+1}-2} + \log_\beta \left(\frac{2^{k+1}-4}{2^{3}-4} \ \frac{2^{3}+2^{k+1}-7}{2^{k+1}+2^{k+1}-7}\right) \\
&+\ \log_\beta \frac{4}{2^{k+1}+1} +\log_\beta \left(\frac{2^{k+1}-2}{2^{2}-2} \ \frac{2^{2}-1}{2^{k+1}-1}\right) +\log_\beta \frac{2}{3}\\[2mm]

=&-1 + \log_\beta  \left(\frac{2^{k+2}-6}{2^{k+1}-2} \ \frac{2^{k+1}-4}{4} \ \frac{2^{k+1}+1}{2^{k+2}-7} \ \frac{4}{2^{k+1}+1} \ \frac{2^{k+1}-2}{2} \ \frac{3}{2^{k+1}-1} \ \frac{2}{3}\right)\\[2mm]
=&-1 + \log_\beta \left(\frac{2^{k+2}-6}{2^{k+2}-7} \frac{2^{k+1}-4}{2^{k+1}-1}\right),
\end{array}\]
which proves Eq.~(\ref{eq:potdiff}).
Note that the leftist heaps $W_k=\Bin{T_k,0,U_k}$ have been crafted to ensure proper cancellations in this result. 

\end{document}